\documentclass[11pt]{article}
\usepackage{graphicx,amssymb,amsmath,amsthm}
\usepackage[linesnumbered, ruled]{algorithm2e}
\SetKwInput{KwResult}{Function}

\usepackage{amsfonts}
\usepackage{booktabs}
\usepackage{cite}
\usepackage{color}
\usepackage{enumerate}
\usepackage{float}
\usepackage{mathrsfs}
\usepackage{subfig}
\usepackage{tabularx}
\usepackage{fullpage}
\usepackage[top=0.8in, bottom=0.9in, left=0.9in, right=0.9in]{geometry}
\usepackage[pdftex, plainpages=false, pdfpagelabels, 
             bookmarks=false,
             bookmarksopen=true,
             bookmarksnumbered=true,
             breaklinks=true,
             linktocpage,
             pagebackref,
             colorlinks=true,
             linkcolor=blue,
             urlcolor=cyan,
             citecolor=red,
             anchorcolor=green,
             hyperindex=true,
             hyperfigures]{hyperref}
\graphicspath{{./Fig/}}

\def\calL{\mathcal{L}}

\def\calA{\mathcal{A}}

\def\calL{\mathcal{L}}
\def\calH{\mathcal{H}}

\def\opta{A_{\text{opt}}}
\def\optp{P_{\text{opt}}}

\def\findmin{FindMinCost}
\def\reset{ResetCost}

\newtheorem{observation}{Observation}

\newtheorem{lemma}{Lemma}
\newtheorem{theorem}{Theorem}

\newtheorem{corollary}{Corollary}
\newtheorem{invariant}{Algorithm Invariant}

\title{Minimum-Weight Half-Plane Hitting Set\thanks{A preliminary version of this paper will appear in {\em Proceedings of the 37th Canadian Conference on Computational Geometry (CCCG 2025)}. This research was supported in part by NSF under Grant CCF-2300356.}
}

 \author{
 Gang Liu\thanks{Kahlert School of Computing,
 University of Utah, Salt Lake City, UT 84112, USA. {\tt gangliu@cs.utah.edu}}
 \and
 Haitao Wang\thanks{Kahlert School of Computing,
 University of Utah, Salt Lake City, UT 84112, USA. {\tt haitao.wang@utah.edu}}
 }

\begin{document}
\pagestyle{plain}
\date{}

\thispagestyle{empty}
\maketitle

\vspace{-0.3in}

\begin{abstract}
Given a set $P$ of $n$ weighted points and a set $H$ of $n$ half-planes in the plane, the hitting set problem is to compute a subset $P'$ of points from $P$ such that each half-plane contains at least one point from $P'$ and the total weight of the points in $P'$ is minimized. The previous best algorithm solves the problem in $O(n^{7/2}\log^2 n)$ time. In this paper, we present a new algorithm with runtime $O(n^{5/2}\log^2 n)$.
\end{abstract}


\section{Introduction}
\label{sec:intro}
Let $H$ be a set of $n$ half-planes and $P$ a set of $n$ points in the plane such that each point has a weight. We say that a point of $P$ {\em hits} a half-plane of $H$ if the half-plane covers the point. A subset $P'$ of $P$ is a {\em hitting set} for $H$ if every half-plane of $H$ is hit by a point in $P'$; $P'$ is a {\em minimum-weight} hitting set if the total weight of all points of $P'$ is the smallest among all hitting sets of $H$. 

In this paper, we consider the {\em half-plane hitting set problem} which is to compute a minimum-weight hitting set for $H$.
The problem has been studied before~\cite{ref:ChanEx14, ref:Har-PeledWe12, ref:LiuGe23, ref:LiuOn24, ref:LiuUn24}. Har-Peled and Lee~\cite{ref:Har-PeledWe12} first proposed an $O(n^6)$ time algorithm. Liu and Wang~\cite{ref:LiuGe23} gave an improved solution by reducing the problem to $O(n^2)$ instances of the {\em lower-only} half-plane hitting set problem, where all half-planes are lower ones. Consequently, if the lower-only problem can be solved in $O(T)$ time, the general problem is solvable in $O(n^2 \cdot T)$ time. Liu and Wang~\cite{ref:LiuGe23} derived an $O(n^2 \log n)$ time algorithm for the lower-only problem, thus solving the general half-plane hitting set problem in $O(n^4 \log n)$ time. More recently, Liu and Wang~\cite{ref:LiuOn24} proposed an improved $O(n^{3/2} \log^2 n)$ time algorithm for the lower-only problem, leading to an $O(n^{7/2} \log^2 n)$ time solution for the general problem.

We present a new algorithm of $O(n^{5/2} \log^2 n)$ time. This improves the previous best solution~\cite{ref:LiuOn24} by a linear factor. More specifically, our algorithm runs in $O(\kappa\cdot n^{{3}/{2}} \log^2 n)$, where $\kappa$ is the minimum number of points of $P$ covered by any half-plane of $H$. As such, our algorithm could be more efficient if $\kappa$ is relatively small in certain applications. For example, 
if there exists a half-plane in $H$ that covers only $O(1)$ points of $P$, then the time complexity of our algorithm is bounded by $O(n^{{3}/{2}} \log^2 n)$, matching the performance of the previous best algorithm for the lower-only case~\cite{ref:LiuOn24}.

\paragraph{\bf Related work.}
In the {\em unweighted} half-plane hitting set problem, all points have the same weight. As it is a special case of the weighted problem, all above results are applicable to the unweighted case. Recently, Liu and Wang proposed an $O(n\log n)$-time algorithm for the unweighted case~\cite{ref:LiuAn25}. Additionally, Wang and Xue~\cite{ref:WangAl24} proved a lower bound of $\Omega(n \log n)$ time under the algebraic decision tree model for the unweighted case, even if all half-planes are lower ones. Therefore, Liu and Wang's algorithm is optimal.

A closely related problem is the half-plane {\em coverage} problem. Given $P$ and $H$ as above with each half-plane associated with a weight, the problem is to compute a minimum-weight subset of half-planes whose union covers all points of $P$. Logan and Wang~\cite{ref:PedersenAl22} reduced this problem to $O(n^2)$ instances of the {\em lower-only} half-plane coverage problem, where all half-planes are lower ones. Consequently, if the lower-only problem can be solved in $O(T')$ time, then the general problem is solvable in $O(n^2 \cdot T')$ time. Notice that the lower-only coverage problem is dual to the lower-only hitting set problem (i.e., the two problems can be reduced to each other in linear time). Therefore, an algorithm for the lower-only hitting set problem can be used to solve the lower-only coverage problem with the same time complexity, and vice versa. In fact, the $O(n^{3/2} \log^2 n)$ time algorithm by Liu and Wang~\cite{ref:LiuOn24} for the lower-only hitting set problem was originally described to solve the lower-only coverage problem. 

It should be noted that in the general case where both upper and lower half-planes are present in $H$, the hitting set problem and the coverage problem are no longer dual to each other. While the unweighted general half-plane hitting set problem can be solved in $O(n\log n)$ time~\cite{ref:LiuAn25}, the currently best algorithm for the unweighted general half-plane coverage problem, which is due to Wang and Xue~\cite{ref:WangAl24}, runs in $O(n^{4/3} \log^{5/3} n \log^{O(1)} \log n)$ time.


The hitting set problem, as a fundamental problem, has been extensively explored in the literature. Various geometric versions of this problem have garnered significant attention, and many are known to be NP-hard~\cite{ref:OualiA14,ref:ChanEx14,ref:MustafaIm10, ref:BusPr18, ref:EvenHi05, ref:GanjugunteGe11, ref:LiA15}. For instance, in the disk hitting set problem, given a set of disks and a set of points in the plane, the goal is to compute a smallest subset of points that hit all the disks. This problem remains NP-hard even when all disks have the same radius~\cite{ref:DurocherDu15, ref:KarpRe72, ref:MustafaIm10}. 
However, certain variants of the problem can be solved in polynomial time. Liu and Wang \cite{ref:LiuGe23} studied a line-constrained version, where the centers of all disks lie on a line, and they developed a polynomial-time algorithm for this case. For related problems in similar geometric settings, see also \cite{ref:LiuOn24,ref:LiuUn24,ref:PedersenOn18,ref:PedersenAl22}. 

\paragraph{\bf Our approach.}
To solve the half-plane hitting set problem on $ P $ and $ H $, instead of reducing the problem to $ O(n^2) $ instances of the lower-only case as in the previous work~\cite{ref:LiuGe23}, we follow a problem reduction method in \cite{ref:LiuAn25} for the unweighted case and reduce the problem to the following \emph{circular-point coverage} problem: Given a set $ \calA $ of weighted arcs and a set $ B $ of points on a circle $ C $, compute a minimum-weight subset of arcs whose union covers all points. 
In our reduction, the points of $B$ are defined by the half-planes of $H$ and the arcs of $\calA$ are defined by the points of $P$. 
While each half-plane of $ H $ defines exactly one  point in $ B $ and thus $ |B| = n $, each point of $ P $ may generate as many as $ n/2 $ (disjoint) arcs in $ \calA $ and thus $ |\calA| = O(n^2)$. Our reduction ensures that a point $ p \in P $ hits a half-plane $ h \in H $ if and only if the point in $ B $ defined by $ h $ is covered by one of the arcs generated by $ p $. An optimal solution to the hitting set problem on $ P $ and $ H $ can be easily obtained from an optimal solution to the circular-point coverage problem: Suppose that $ \opta $ is a minimum-weight subset of arcs in $ \calA $ whose union covers $ B $; for each arc in $ \opta $, if the arc is generated by a point $ p \in P $, then we add $ p $ to $ \optp $. We prove that the subset $ \optp $ thus obtained is a minimum-weight hitting set for $ H $ (note that this implies that no two arcs from $\opta$ are generated by the same point).

To solve the circular-point coverage problem, while the unweighted case of this problem can be solved in $O((|\calA| + |B|) \log (|\calA| + |B|))$ time~\cite{ref:WangAl24}, 
the currently best algorithm for the weighted case, which is due to Atallah, Chen, and Lee \cite{ref:AtallahAn95}, runs in $O(\kappa\cdot (|\calA| + |B|) + (|\calA| + |B|) \log (|\calA| + |B|))$ time, where $\kappa$ is the minimum number of arcs covering any point of $B$. If applied to our problem, $\kappa$ is equal to the minimum number of points of $P$ covered by any half-plane of $H$, and thus $\kappa\leq n$ and applying the algorithm of \cite{ref:AtallahAn95} can solve our problem in $O(\kappa\cdot n^2+n^2\log n)$ time, which is  $O(n^3)$ in the worst case. 

We take a different route to tackle the problem. We further reduce the problem to $\kappa$ instances of an interval coverage problem. In each instance, we have $O(n^2)$ weighted intervals and $O(n)$ points on a line, and the goal is to compute a minimum-weight subset of intervals whose union covers all points. Clearly, solving each problem instance requires $\Omega(n^2)$ time as the number of intervals could be $\Omega(n^2)$. We instead solve each problem instance in an ``indirect'' way and the ``indirect'' solutions of all $\kappa$ problem instances together also lead to an optimal solution to our original circular-point coverage problem. The advantage of computing an indirect solution is that we do not have to explicitly have the $O(n^2)$ intervals and instead only need to use the original input $P$ and $H$. Our algorithm can compute an indirect solution for each problem instance in $O(n^{3/2}\log^2 n)$ time. Consequently, we can solve our circular-point coverage problem and thus the hitting set problem in a total of $O(\kappa\cdot n^{3/2}\log^2 n)$ time.

\paragraph{\bf Outline.} The rest of the paper is organized as follows. After defining the notation in Section~\ref{sec:pre}, we reduce the problem to the circular-point coverage problem in Sections~\ref{sec:trans} and then solve it in Section~\ref{sec:trans2}. As discussed above, we solve the circular-point coverage problem by reducing it to $\kappa$ instances of an interval coverage problem. Our algorithm for the interval coverage problem is also described in Section~\ref{sec:trans2}, but the implementation of the algorithm is discussed in Section~\ref{sec:implesketch}. 

\section{Preliminaries}
\label{sec:pre}
We follow the notation in Section~\ref{sec:intro}, e.g., $P$, $H$, $n$, etc. We assume that each half-plane must be hit by a point of $P$; otherwise, a hitting set does not exist. We can check whether this condition is satisfied in $O(n\log n)$ time (e.g., by first computing the convex hull of $P$, and then for each half-plane $h\in H$, by using the convex hull we can determine whether $h$ contains a point of $P$ in $O(\log n)$ time). For each point $p\in P$, let $w(p)$ represent the weight of $p$; we assume $w(p)>0$ since otherwise $p$ could always be included in the solution.

For each half-plane $h\in H$, its \emph{normal} is the vector perpendicular to its bounding line and toward the interior of $h$. If two half-planes $h, h' \in H$ have the same normal, then one of them must contain the other, say, $h \subseteq h'$. In this case, $h'$ is redundant because any point hitting $h$ also hits $h'$. We can efficiently identify redundant half-planes in $O(n \log n)$ time by sorting the half-planes by their normals. In the following, we assume that half-planes of $H$ have distinct normals. 


\section{Reducing to a circular-point coverage problem}
\label{sec:trans}

In this section, we reduce the half-plane hitting set problem for $P$ and $H$ to a circular-point coverage problem for a set $\calA$ of arcs and a set $B$ of points on a circle $C$. In what follows, we first define the circular-point coverage problem, i.e., $\calA$, $B$, and $C$. We then prove the correctness of the reduction, i.e., explain how and why a solution to the circular-point coverage problem can lead to a solution to the half-plane hitting set problem. 
Our reduction is similar to the approach in \cite{ref:LiuAn25} for the unweighted case. The difference is that here we have to consider the weights of the points. One may consider our approach an extension of that in \cite{ref:LiuAn25}. 

\subsection{Defining the circular-point coverage problem}

First, we define $C$ as a unit circle.

\paragraph{Defining $\boldsymbol{B}$.} For each half-plane $h \in H$, we define a point $b$ on $C$ as the intersection of $C$ with the ray originating from the center of $C$ and parallel to the normal of $h$. We call $h$ the {\em defining half-plane} of $b$. Let $B$ be the set of all such points on $C$ defined by the half-planes of $H$. Since no two half-planes in $H$ have the same normal, points in $B$ are distinct on $C$.


Let $b_1, b_2, \ldots, b_n$ be the points of $B$ ordered counterclockwise on $C$; we consider it a cyclic list. We use $B[i,j]$ to refer to the (contiguous) sublist of $B$ counterclockwise from $b_i$ to $b_j$, inclusive, i.e., if $i \leq j$, then $B[i,j] = {b_i, b_{i+1}, \ldots, b_j}$; otherwise,  $B[i,j] = {b_i, \ldots, b_n, b_1, \ldots, b_j}$. For any two points $b$ and $b'$ on $C$, let $C[b,b']$ denote the arc of $C$ counterclockwise from $b$ to $b'$, inclusive.


\paragraph{\bf Defining $\boldsymbol{\calA}$.} For each point $p \in P$, we define a set $A(p)$ of arcs on $C$. 
For each maximal sublist $B[i,j]$ of $B$ such that $p$ hits all defining half-planes of the points in $B[i,j]$, we add the arc $C[b_i,b_j]$ to $A(p)$. In the special case where $p$ hits all the half-planes of $H$, we let $A(p)$ consist of the single arc $C[b_1,b_n]$. 
Note that $A(p)$ may contain at most $\lfloor n/2 \rfloor$ arcs, and all these arcs are pairwise disjoint. We say that $p$ is the {\em defining point} of the arcs of $A(p)$. For each arc $\alpha \in A(p)$, we let $w(\alpha)=w(p)$ denote its weight. Define $\calA=\bigcup_{p\in P} A(p)$. Clearly, $|\calA| = O(n^2)$.

\subsection{Correctness of the reduction}

Consider the circular-point coverage problem for $B$ and $\calA$. Suppose $\opta$ is an optimal solution, i.e., $\opta$ is a minimum-weight subset of $\calA$ for covering $B$. We create a subset $\optp$ of $P$ as follows: For each arc in $\opta$, we add its defining point to $\optp$. 

In the following, we prove that $\optp$ is an optimal solution to the half-plane hitting set problem for $P$ and $H$. At first glance, one potential issue is that two arcs of $\opta$ might be defined by the same point of $P$ and thus a point might be added to $\optp$ multiple times. We will show that this is not possible, implying that $\sum_{p \in \optp} w(p) = \sum_{\alpha \in \opta} w(\alpha)$. These are proved in Corollary~\ref{corollary:circular-point}, which follows mostly from the following lemma.

\begin{lemma}
\label{lemm:circular-point}
If there exists a subset of arcs of $\calA$ with total weight $W$ that forms a coverage for $B$, then there exists a subset of $P$ with total weight at most $W$ that forms a hitting set for $H$. Symmetrically, if there exists a subset of $P$ with total weight $W$ that forms a hitting set for $H$, then there exists a subset of arcs of $\calA$ with total weight at most $W$ that forms a coverage for $B$. 
\end{lemma}
\begin{proof}
Suppose $B$ can be covered by a subset $A' \subseteq \calA$ of arcs with total weight $W$. Then, let $P'$ be the set of defining points of all arcs of $A'$. Clearly, $\sum_{p \in P'} w(p) \leq \sum_{\alpha \in A'} w(\alpha) = W$. We claim that $P'$ is a hitting set for $H$. To see this, consider a half-plane $h\in H$. Let $b$ be the point of $B$ defined by $h$. Then, $b$ must be covered by an arc $\alpha$ of $A'$. By definition, $h$ is hit by the defining point of $\alpha$, which is in $P'$. Therefore, $P'$ is a hitting set for $H$. This proves one direction of the lemma. In the following, we prove the other direction.

Suppose $H$ has a hitting set $P'\subseteq P$ and the total weight of the points of $P'$ is $W$. Our goal is to show that there exists a subset $A'\subseteq \calA$ of arcs such that their union covers $B$ and $\sum_{\alpha \in A'} w(\alpha) \leq W$. Depending on the size $|P'|$, there are three cases.

\begin{enumerate}
    \item If $|P'|=1$, then let $p$ be the only point of $P'$. By definition, $p$ hits all half-planes of $H$ and $W = w(p)$. Hence, $p$ defines a single arc $\alpha=C[b_1,b_n]$ in $\calA$ and $w(\alpha)=w(p)$. Clearly, the arc $\alpha$ covers all points of $B$. As such, if we have $A' = \{\alpha\}$, then $A'$ satisfies our condition.
    
    \item If $|P'|=2$, then let $p_u$ and $p_l$ denote the two points of $P'$, respectively. Hence, $W = w(p_u) + w(p_l)$. We rotate the coordinate system so that the line segment $\overline{p_up_l}$ is vertical with $p_u$ higher than $p_l$. Since every half-plane of $H$ is hit by $p_u$ or $p_l$, it is not difficult to see that all upper half-planes (in the rotated coordinate system) must be hit by $p_u$ and all lower half-planes must be hit by $p_l$. Hence, $p_u$ must define an arc $\alpha_u$ in $\calA$ that covers all points of $B$ that are defined by all upper half-planes; similarly, $p_l$ must define an arc $\alpha_l$ in $\calA$ that covers all points of $B$ that are defined by all lower half-planes. Therefore, $\alpha_u$ and $\alpha_l$ together cover all points of $B$. If we let $A' = \{\alpha_l, \alpha_u\}$, we have $\sum_{\alpha \in A'} w(\alpha) = w(\alpha_l) + w(\alpha_u) = w(p_l) + w(p_u) = W$.

    \item If $|P'| \geq 3$, then let $p_1, p_2, \ldots, p_t$ be vertices of the convex hull of $P'$, ordered counterclockwise. Clearly, $t\geq 3$ and $\sum_{1\leq i\leq t}w(p_i)\leq W$.  For each pair of adjacent points $p_i$ and $p_{i+1}$ on the convex hull, with indices modulo $t$, define $h_{i,i+1}$ as the half-plane whose bounding line contains the line segment $\overline{p_i p_{i+1}}$ such that it does not contain the interior of the convex hull of $P'$ (see Fig.~\ref{fig:convexhullweighted}). We define a point $b'_{i,i+1}$ on the circle $C$ using the normal of $h_{i,i+1}$. The $t$ points $b'_{i,i+1}$, $1 \leq i \leq t$, partition $C$ into $t$ arcs $\alpha_i'$, where $\alpha_i' = C[b'_{i-1,i}, b'_{i,i+1}]$. In the following, we show that for each arc $\alpha_i'$, the point $p_i$ defines an arc in $\calA$ that covers all points of $B \cap \alpha_i'$. 
    
\begin{figure}[h]
\centering
\includegraphics[height=1.7in]{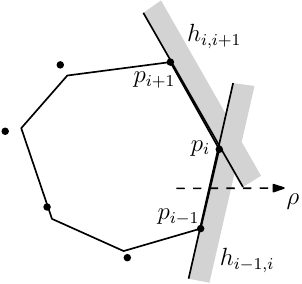}
\caption{Illustration of $h_{i-1,i}$, $h_{i,i+1}$, and $\rho$.}
\label{fig:convexhullweighted}
\end{figure}

    Consider a point $b \in \alpha_i'$, let $h\in H$ be its defining half-plane. Let $\rho$ denote the normal of $h$. By definition, $\rho$ is in the interval of directions counterclockwise from the normal of $h_{i-1,i}$ to that of $h_{i,i+1}$ (see Fig.~\ref{fig:convexhullweighted}). It is not difficult to see that the point $p_i$ is the most extreme point of $P'$ along the direction $\rho$. Since $h$ is hit by a point of $P'$, $h$ must be hit by $p_i$. This means that $b$ must be covered by an arc defined by $p_i$. As all points of $\alpha_i' \cap B$ are contiguous in the cyclic list of $B$, they must be covered by a single arc defined by $p_i$; let $\alpha_i$ denote the arc. Note that $w(\alpha_i)=w(p_i)$. 

    As the union of arcs $\alpha'_i$, $1\leq i\leq t$, covers $C$, the union of arcs $\alpha_i$, $1\leq i\leq t$, covers $B$. Let $A'$ be the set of the arcs $\alpha_i$, $1\leq i\leq t$. We have $\sum_{\alpha \in A'} w(\alpha) = \sum_{1 \leq i \leq t} w(\alpha_i) = \sum_{1 \leq i \leq t} w(p_i) \leq W$.   
\end{enumerate}

The lemma thus follows. 
\end{proof}


\begin{corollary}
\label{corollary:circular-point}
\begin{enumerate}
    \item The total weight of a minimum-weight subset of $\calA$ for covering $B$ is equal to the total weight of a minimum-weight subset of $P$ for hitting $H$.
    \item No two arcs of $\opta$ are defined by the same point.
    \item $\optp$ is an optimal solution to the half-plane hitting set problem.
\end{enumerate}
\end{corollary}

\begin{proof}
The first corollary statement directly follows from Lemma~\ref{lemm:circular-point}.

For the second corollary statement, notice that $\optp$ is a hitting set of $H$, which follows a similar argument to the first direction of Lemma~\ref{lemm:circular-point}. Assume to the contrary that $\opta$ has two arcs $\alpha$ and $\alpha'$ defined by the same point of $P$. Then, since the weight of every arc is positive, by the definition of $\optp$, $\sum_{p \in \optp} w(p) < \sum_{\alpha \in \opta} w(\alpha)$ holds. As $\opta$ is a minimum-weight subset of $\calA$ for covering $B$ and $\optp$ is a hitting set of $H$, we obtain that the total weight of a hitting set for $H$ is smaller than the total weight of a minimum-weight subset of $\calA$ for covering $B$, a contradiction to the first corollary statement.

For the third statement, due to the second statement, by the definition of $\optp$, it holds that 
$\sum_{p \in \optp} w(p) = \sum_{\alpha \in \opta} w(\alpha)$. As $\opta$ is a minimum-weight subset of $\calA$ for covering $B$, and $\optp$ is a hitting set for $H$, by the first corollary statement, $\optp$ must be a minimum-weight hitting set for $H$.
\end{proof}

By Corollary~\ref{corollary:circular-point}, we have successfully reduced our half-plane hitting set problem on $H$ and $P$ to the circular-point coverage problem on $B$ and $\calA$. 

\section{Solving the circular-point coverage problem}
\label{sec:trans2}

In this section, we present an algorithm to solve the circular-point coverage problem in Section~\ref{sec:trans}, i.e., computing a subset of arcs of $\calA$ such that their union covers all points of $B$ and their total weight is minimized (such a subset is called an {\em optimal solution} of the problem; we often use $\opta$ to denote an optimal solution). Recall that $|\calA|=O(n^2)$ and $|B|=n$. 

\subsection{Reducing to interval coverage}

Our strategy is to reduce the problem to multiple instances of an interval coverage problem and then solve all instances, whose solutions together will lead to a solution to the circular problem. 

We start by finding a point $b^*\in B$ that is covered by the fewest arcs in $\calA$. This can be easily done in $O(n^2\log n)$ time by sorting the endpoints of the arcs of $\calA$ along with the points of $B$. 
We have the following easy observation. 

\begin{observation}\label{obser:10}
For any optimal solution $\opta$ to the circular-point coverage problem, $b^*$ must be covered by an arc $\alpha\in \opta$ and $\opta$ does not have any other arc that fully contains $\alpha$. 
\end{observation}
\begin{proof}
    As $\opta$ is an optimal solution, all points of $B$ are covered by the union of the arcs of $\opta$. Therefore, $b^*$ must be covered by an arc $\alpha\in \opta$. Assume to the contrary that $\opta$ has another arc $\alpha'$ that fully contains $\alpha$. Then, if we remove $\alpha$ from $\opta$, the remaining arcs of $\opta$ still form a coverage of $B$. But this contradicts the optimality of $\opta$. 
\end{proof}

For each arc $\alpha\in \calA$ that covers $b^*$, define $B_{\alpha}$ as the subset of points of $B$ that are not covered by $\alpha$, and define $\calA_{\alpha}$ as the subset of arcs of $\calA\setminus\{\alpha\}$ that do not fully contain $\alpha$ and are not defined by the point $p({\alpha})$, where $p({\alpha})$ is the point of $P$ that defines $\alpha$. 

We consider the following {\em interval coverage problem for $\alpha$}: 
Find a subset of arcs of $\calA_{\alpha}$ such that their union covers $B_{\alpha}$ and the total weight of the arcs is minimized. 
This is a 1D interval coverage problem (not a circular one). Indeed, let $b_{ccw}$ (resp., $b_{cw}$) be the first point of $B$ counterclockwise from $b^*$ not covered by $\alpha$ (see Figure~\ref{fig:interval_coverage10}). By definition, no arc of $\calA_{\alpha}$ fully contains $\alpha$. Hence, no arc of $\calA_{\alpha}$ contains all three points $b_{cw}$, $b^*$, and $b_{ccw}$. This means that we can break the circle $C$ at $b^*$ into an arc, denoted by $C_{b^*}$, such that every arc of $\calA_{\alpha}$ is contained in $C_{b^*}$. Furthermore, we can map $C_{b^*}$ to the $x$-axis $\ell$ from left to right containing all points of $B_{\alpha}$ from $b_{ccw}$ counterclockwise to $b_{cw}$, so that each arc of $\calA_{\alpha}$ becomes an interval of $\ell$. Therefore, the problem is essentially a 1D interval coverage problem. Let $A^*_{\alpha}$ denote an optimal solution of the problem. 

\begin{figure}[t]
\begin{center}
\includegraphics[height=1.7in]{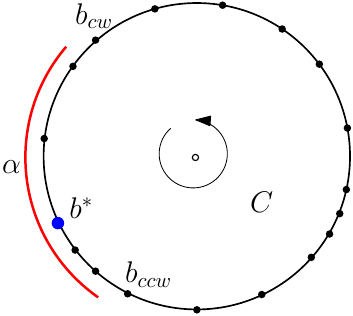}
\caption{Illustration of the definitions of $b_{ccw}$ and $b_{cw}$.}
\label{fig:interval_coverage10}
\end{center}
\vspace{-0.15in}
\end{figure}

For any subset $\calA'$ of arcs of $\calA$, we use $W(\calA')$ denote the total weight of all arcs of $\calA'$. 

By slightly abusing the notation, let $\calA_{b^*}$ denote the subset of arcs of $\calA$ that contain $b^*$. 
The following lemma explains why we are interested in the above interval coverage problem. 

\begin{lemma}\label{lem:20}
Suppose that $\alpha'=\text{\em argmin}_{\alpha\in \calA_{b^*}} (w(\alpha)+W(A^*_{\alpha}))$, i.e., 
among all arcs $\alpha\in \calA_{b^*}$, $\alpha'$ is the one with the minimum $w(\alpha)+W(A^*_{\alpha})$. Then, $\{\alpha'\}\cup A^*_{\alpha'}$ is an optimal solution to the circular-point coverage problem for $\calA$ and $B$.
\end{lemma}
\begin{proof}
Let $\opta$ be an optimal solution to the circular-point coverage problem for $\calA$ and $B$. By definition, the arcs of $\{\alpha'\}\cup A^*_{\alpha'}$ form a coverage of $B$. Since $\opta$ is an optimal solution, we have $W(\opta)\leq W(\{\alpha'\}\cup A^*_{\alpha'})=w(\alpha')+W(A^*_{\alpha'})$. In the following, we show that $w(\alpha')+W(A^*_{\alpha'})\leq W(\opta)$, which will prove the lemma. 

By Observation~\ref{obser:10}, $b^*$ is covered by an arc $\alpha\in \opta$ and $\opta$ does not have another arc that fully contains $\alpha$. By Corollary~\ref{corollary:circular-point}, no two arcs of $\opta$ are defined by the same point of $P$. Consequently, by the definition of $\calA_{\alpha}$, we have $\opta\setminus\{\alpha\}\subseteq \calA_{\alpha}$. Since $\alpha$ does not cover any point of $B_{\alpha}$ by definition, $\opta\setminus\{\alpha\}$ must form a coverage of $B_{\alpha}$, and thus $W(A^*_{\alpha})\leq W(\opta\setminus\{\alpha\})$. As such, by the definition of $\alpha'$, we obtain $w(\alpha')+W(A^*_{\alpha'})\leq w(\alpha) + W(A^*_{\alpha}) \leq w(\alpha)+W(\opta\setminus\{\alpha\})=W(\opta)$.
\end{proof}


With Lemma~\ref{lem:20}, it is attempting to solve the interval coverage problem for all arcs $\alpha\in \calA_{b^*}$. However, this would not be efficient because solving the problem for each arc $\alpha$ takes $\Omega(n^2)$ time due to  $|\calA_{\alpha}|=\Omega(n^2)$ in the worst case (and thus the overall time complexity would be $\Omega(n^3)$ in the worst case). Instead, we take an ``indirect'' strategy as follows. 


\paragraph{An indirect route.}
Let $H_{\alpha}$ denote the set of half-planes of $H$ that define the points of $B_{\alpha}$. Let $P_{\alpha}=P\setminus\{p({\alpha})\}$, where $p({\alpha})$ is the point of $P$ that defines $\alpha$. 
Given the arc $\alpha$, $B_{\alpha}$ and $P_{\alpha}$ can be easily computed in $O(n)$ time. Clearly, $|H_{\alpha}|=|B_{\alpha}|\leq n-1$ and $|P_{\alpha}|=n-1$.

For each arc $\alpha\in \calA_{b^*}$, our strategy is to design an efficient algorithm to compute  a subset $P'_{\alpha}\subseteq P_{\alpha}$ and a value $W_{\alpha}$ such that $P'_{\alpha}$ is a hitting set for $H_{\alpha}$
and $W(P'_{\alpha})\leq W_{\alpha}\leq  W(A^*_{\alpha})$. We refer to $(P'_{\alpha},W_{\alpha})$ as an ``indirect solution'' to the interval coverage problem for $\alpha$. 
The following lemma explains why such indirect solutions are sufficient for our purpose. 

\begin{lemma}\label{lem:optweightedcircle}
Suppose that $\alpha^*=\text{\em argmin}_{\alpha\in \calA_{b^*}} (w(\alpha)+W_{\alpha})$.
Then $\{p(\alpha^*)\}\cup P'_{\alpha^*}$ is an optimal solution to the half-plane hitting set problem for $P$ and $H$, where $p(\alpha^*)$ is the point of $P$ that defines $\alpha^*$.
\end{lemma}
\begin{proof}
Let $\optp$ be an optimal solution to the half-plane hitting set problem, and $\opta$ an optimal solution to the circular-point coverage problem for $\calA$ and $B$. By Corollary~\ref{corollary:circular-point}, $W(\optp)=W(\opta)$. 

Since all points of $B\setminus B_{\alpha^*}$ are covered by the arc $\alpha^*$, by definition, all half-planes of $H\setminus H_{\alpha^*}$ are hit by $p(\alpha^*)$. Since each half-plane of $H_{\alpha^*}$ is hit by a point of $P'_{\alpha^*}$, we obtain that $\{p(\alpha^*)\}\cup P'_{\alpha^*}$ is a hitting set for $H$. As such, it holds that $W(\optp)\leq w(p(\alpha^*))+W(P'_{\alpha^*})=w(\alpha^*)+W(P'_{\alpha^*})\leq w(\alpha^*)+W_{\alpha^*}$. In the following, we argue that $w(\alpha^*)+W_{\alpha^*}\leq W(\optp)$, which will lead to  $W(\optp)= w(p(\alpha^*))+W(P'_{\alpha^*})$ and thus prove the lemma.

Suppose that $\alpha'=\text{argmin}_{\alpha\in \calA_{b^*}} (w(\alpha)+W(A^*_{\alpha}))$. 
By Lemma~\ref{lem:20}, $W(\opta)=w(\alpha')+W(A^*_{\alpha'})$. By the definition of $\alpha^*$, $w(\alpha^*)+W_{\alpha^*}\leq w(\alpha')+W_{\alpha'}$. Since $W_{\alpha'}\leq W(A^*_{\alpha'})$, we obtain $w(\alpha^*)+W_{\alpha^*}\leq w(\alpha')+W(A^*_{\alpha'})=W(\opta)=W(\optp)$.
\end{proof}

\subsection{Computing the interval coverage indirection solutions}

We now present our algorithm to compute an indirect solution $(P'_{\alpha},W_{\alpha})$ for the interval coverage problem for each $\alpha$; the runtime of the algorithm is $O(n^{3/2}\log^2 n)$. In the following, we first describe the algorithm, then prove its correctness, and finally discuss how to implement it efficiently.

\subsubsection{Algorithm description}
\label{sec:description}

Recall that $|H_{\alpha}|=|B_{\alpha}|\leq n-1$ and $|P_{\alpha}|=n-1$. For notational convenience, we simple let $|H_{\alpha}|=|B_{\alpha}|=|P_{\alpha}|=n$. 

We first sort the points of $B_{\alpha}$ counterclockwise along $C_{\alpha}$ as $b_1,b_2,\ldots,b_n$. 
For each $b_i$, let $h_i$ denote its defining half-plane in $H_{\alpha}$. For each half-plane $h$, denote by $P_{\alpha}(h)$ the subset of points of $P_{\alpha}$ inside $h$, and $P_{\alpha}(\overline{h})$ the subset of points of $P_{\alpha}$ outside $h$.

Our algorithm processes the half-planes of $H_{\alpha}$ in their index order. For each half-plane $h_i\in H_{\alpha}$, the algorithm computes a value $\delta_i$. The algorithm also maintains a value $cost(p)$ for each point $p\in P_{\alpha}$, which is initialized to its weight $w(p)$. The pseudocode of the algorithm is given in Algorithm~\ref{algo:10}. 

\begin{algorithm}[h]
	\caption{}
	\label{algo:10}
	\SetAlgoNoLine
	\KwIn{$P_{\alpha}$ and $H_{\alpha}=\{h_1,h_2,\ldots,h_n\}$  
 }
	\KwOut{An indirect solution $(P'_{\alpha},W_{\alpha})$} \BlankLine
	$cost(p)\leftarrow w(p)$, for all points $p\in P_{\alpha}$\;
	\For{$i\leftarrow 1$ \KwTo $n$}
	{
        $\delta_i\leftarrow \min_{p\in P_{\alpha}(h_i)} cost(p)$\tcp*[r]{\findmin\ Operation}         
        \label{ln:findmin}
        $cost(p)\leftarrow w(p) + \delta_i$ for all points $p\in P_{\alpha}(\overline{h_i})$\tcp*[r]{\reset\ Operation} 
        \label{ln:reset}
	}
    \Return $\delta_n$ as $W_{\alpha}$\tcp*[r]{$P'_{\alpha}$ consists of the points $p$ whose weights $w(p)$ are included in $\delta_n$ and $P'_{opt}$ can be found by the standard backtracking technique} 
\end{algorithm}

The algorithm is essentially a dynamic program. We will prove later that the value $W_{\alpha}$ returned by the algorithm, which is $\delta_n$, is smaller than or equal to $W(A^*_{\alpha})$. According to the algorithm, $\delta_n$ is equal to the total weight of a subset of points of $P_{\alpha}$; we let $P'_{\alpha}$ be the subset. To find $P'_{\alpha}$, we can slightly modify the algorithm following the standard dynamic programming backtracking technique. Specifically, if $\delta_n$ is equal to $cost(p)$ for some point $p\in P_{\alpha}(h_n)$, then $p$ hits $h_n$ and we add $p$ to $P'_{\alpha}$. Suppose that $cost(p)$ is equal to $w(p)+\delta_i$ for some index $i$. Then, $cost(p)$ does not change in the $j$-th iteration for any $i<j<n-1$, meaning that $p\in P_{\alpha}(h_j)$. Hence, $p$ hits all half-planes $h_{i+1},\ldots,h_n$. Further, by definition $\delta_i$ is equal to $cost(p')$ for some point $p'\in P_{\alpha}(h_i)$; we add $p'$ to $P'_{\alpha}$. We continue this backtracking process until a point whose cost is equal to its own weight is added to $P'_{\alpha}$ (in which case following the above analysis $P'_{\alpha}$ is a hitting set for $H_{\alpha}$). Note that in the above process if a point $p$ is added to $P'_{\alpha}$ multiple times, then $w(p)$ is included in $\delta_n$ for each copy of $p$. As such, since $W(P'_{\alpha})$ represents the total weight of all distinct points of $P'_{\alpha}$, it holds that $W(P'_{\alpha})\leq \delta_n=W_{\alpha}$. To prove the correctness of the algorithm, it now remains to prove that $W_{\alpha}\leq W(A^*_{\alpha})$, which will be done in Section~\ref{sec:correct}.

For reference purposes, we use \findmin\ to refer to the operation in Line~\ref{ln:findmin} and use \reset\ to refer to the operation in Line~\ref{ln:reset} of Algorithm~\ref{algo:10}. The efficiency of the algorithm hinges on how to implement these two {\em key operations}, which will be discussed later in Section~\ref{sec:imple}.

\subsubsection{Algorithm correctness: $\boldsymbol{W_{\alpha}\leq W(A^*_{\alpha})}$}
\label{sec:correct}
As discussed above, it remains to prove that $W_{\alpha}\leq W(A^*_{\alpha})$. To this end, we use an ``algorithmic proof'': 
We first show that an optimal solution $A^*_{\alpha}$ can be found by a dynamic programming algorithm and then use the algorithm to argue that $W_{\alpha}\leq W(A^*_{\alpha})$.

\paragraph{A dynamic programming algorithm.}
With $B_{\alpha}$ and $\calA_{\alpha}$, computing $A^*_{\alpha}$ can be done by a straightforward dynamic programming algorithm~\cite{ref:PedersenAl22}. To prove  $W_{\alpha}\leq W(A^*_{\alpha})$, we first describe the algorithm. 

Recall that the points of $B_{\alpha}$ are ordered along $C_{\alpha}$ as $b_1,b_2,\ldots,b_n$. As discussed in Section~\ref{sec:intro}, we can map them on the $x$-axis $\ell$ so that $b_1$ (resp., $b_n$) is the leftmost (resp., rightmost) point of $B_{\alpha}$ and each arc of $\calA_{\alpha}$ becomes an interval on $\ell$. 
For convenience, we use $b_0$ to refer to a point of $\ell$ to the left of $b_1$ so that no arc of $\calA_{\alpha}$ contains it. 

For each arc $\beta\in \calA_{\alpha}$, define $I_{\beta}$ as the index of the rightmost point of $B_{\alpha}\cup\{b_0\}$ strictly to the left of the left endpoint of $\beta$. With the point $b_0$, $I_{\beta}$ is well defined. For $1\leq i\leq j\leq n$, define $B_{\alpha}[i,j]=\{b_i,b_{i+1},\ldots,b_j\}$ and $H_{\alpha}[i,j]=\{h_i,h_{i+1},\ldots,h_j\}$.

For each $i\in [1,n]$, define $\delta_i^*$ as the minimum total weight of a subset of arcs of $\calA_{\alpha}$ whose union covers all points of $B_{\alpha}[1,i]$. The goal is to compute $\delta_n^*$, which is equal to $W(A^*_{\alpha})$. For convenience, we let $\delta^*_0=0$.
For each arc $\beta\in \calA_{\alpha}$, define $cost(\beta)=w(\beta)+\delta^*_{I_{\beta}}$. 
One can verify that the following holds: $\delta_i^*=\min_{\beta\in \calA_{\alpha}(b_i)}cost(\beta)$, where $\calA_{\alpha}(b_i)$ is the subset of arcs of $\calA_{\alpha}$ that cover $b_i$. This is the recursive relation of the dynamic program.

We sweep a point $q$ on $\ell$ from left to right. During the sweep, the subset $\calA_{\alpha}(q)\subseteq \calA_{\alpha}$ of arcs that cover $q$ is maintained. We assume that the cost of each arc of $\calA_{\alpha}(q)$ is already known and the values $\delta_i^*$ for all points $b_i\in B_{\alpha}$ to the left of $q$ have been computed. An event happens when $q$ encounters an endpoint of an arc of $\calA_{\alpha}$ or a point of $B_{\alpha}$. 
If $q$ encounters a point $b_i\in B$, then we find the arc of $\calA_{\alpha}(q)$ with the minimum cost and assign the cost value to $\delta_i^*$. If $q$ encounters the left endpoint of an arc $\beta$, we set $cost(\beta)=w(\beta)+\delta^*_{I_{\beta}}$ (assuming that the index $I_{\beta}$ is already knonw) and insert $\beta$ into $\calA_{\alpha}(q)$. If $q$ encounters the right endpoint of an arc $\beta$, we remove $\beta$ from $\calA_{\alpha}(q)$. The algorithm finishes once $q$ encounters $b_n$, at which event $\delta_n^*$ is computed. 

\paragraph{Proving $\boldsymbol{W_{\alpha}\leq W(A^*_{\alpha})}$.}
To prove $W_{\alpha}\leq W(A^*_{\alpha})$, since $W_{\alpha}=\delta_n$ and $W(A^*_{\alpha})=\delta^*_n$, 
it suffices to show that $\delta_n\leq \delta^*_n$. In the following, we argue that $\delta_i\leq \delta_i^*$ holds for all $1\leq i\leq n$. We do so by induction. 

As the base case, we first argue $\delta_1\leq \delta_1^*$. By definition, $\delta_1=\min_{p\in P_{\alpha}(h_1)} w(p)$. For $\delta_1^*$, since $I_{\beta}=0$ for every arc $\beta\in \calA_{\alpha}(b_1)$ and $\delta^*_0=0$, we have $\delta_1^*= \min_{\beta\in \calA_{\alpha}(b_1)} w(\beta)$. 
By definition, an arc $\beta\in \calA_{\alpha}$ covers $b_1$ only if the point $p(\beta)$ of $P_{\alpha}$ defining $\beta$ hits $h_1$, and $w(p(\beta))=w(\beta)$. Therefore, $\beta$ is in $\calA_{\alpha}(b_1)$ only if $p(\beta)$ is in $P_{\alpha}(h_1)$. This implies that $\delta_1\leq \delta_1^*$. 

Consider any $i$ with $2\leq i\leq n$. Assuming that $\delta_{j}\leq \delta^*_{j}$ for all $1\leq j<i$, we now prove $\delta_i\leq \delta^*_i$. Recall that $\delta_i= \min_{p\in P_{\alpha}(h_i)} cost(p)$ and $\delta^*_i= \min_{\beta\in \calA_{\alpha}(b_i)} cost(\beta)$. 
As argued above, each arc $\beta\in \calA_{\alpha}(b_i)$ is defined by a point in $P_{\alpha}(h_i)$ with the same weight. 

Consider an arc $\beta\in \calA_{\alpha}(b_i)$. Let $p$ be the point of $P_{\alpha}(h_i)$ that defines $\beta$. 
To prove $\delta_i\leq \delta_i^*$, it suffices to show that $cost(p)\leq cost(\beta)$. 
By definition, $cost(\beta)=w(\beta)+\delta^*_{I_{\beta}}$. For notational convenience, let $j=I_{\beta}$. By definition, all points of $B_{\beta}[j+1,i]$ are covered by $\beta$ but $b_j$ is not. Therefore, all half-planes of $H_{\alpha}[j+1,i]$ are hit by $p$ but $h_j$ is not. As such, during the \reset\ operation of the $j$-th iteration of Algorithm~\ref{algo:10}, $cost(p)$ will be set to $w(p)+\delta_j$; furthermore, $cost(p)$ will not be reset again during the $i'$-th iteration for all $j+1\leq i'\leq i$. Hence, we have $cost(p)=w(p)+\delta_j$ at the beginning of the $i$-th iteration of the algorithm. Since $\delta_j\leq \delta^*_j$ holds by the induction hypothesis and $w(p)=w(\beta)$, we obtain $cost(p)\leq cost(\beta)$. This proves $\delta_i\leq \delta_i^*$. 

This proves $W_{\alpha}\leq W(A^*_{\alpha})$. The correctness of Algorithm~\ref{algo:10} is thus established. 

\subsubsection{Algorithm implementation}
\label{sec:imple}

We are able to implement the two key operations \findmin\ and \reset\ to achieve an overall time complexity of $O(n^{{3}/{2}} \log^2 n)$ for the entire Algorithm~\ref{algo:10}. The implementation is similar to that for an algorithm of Liu and Wang \cite{ref:LiuOn24} for a line separable unit-disk coverage problem. More specifically, given a set of points and a set of weighted unit disks such that the points are separated from the disk centers by the $x$-axis $\ell$, their problem is to compute a minimum-weight subset of disks whose union covers all points. As such, to cover points, their algorithm uses unit disks (whose centers are all below $\ell$) while ours uses half-planes. We can essentially use the same way as theirs to implement our algorithm with the following changes. Their algorithm initially constructs a cutting~\cite{ref:ChazelleCu93,ref:WangUn23} on the arcs of the boundaries of the disks above the line $\ell$. Instead, we build a cutting on the bounding lines of the half-planes of $H_{\alpha}$. For the rest of the algorithm, we simply replace their disks with our half-planes. The analysis and time complexities are very similar. For completeness, we sketch the main idea in Section~\ref{sec:implesketch}. 

\begin{lemma}\label{lem:indirect}
An indirect solution $(P'_{\alpha},W_{\alpha})$ for the interval coverage problem for each $\alpha\in \calA_{b^*}$ can be computed in $O(n^{{3}/{2}} \log^2 n)$ time. 
\end{lemma}

\subsection{Putting it all together}

Combining Lemmas~\ref{lem:optweightedcircle} and \ref{lem:indirect}, the hitting set problem on $P$ and $H$ can be solved in $O(\kappa\cdot n^{{3}/{2}} \log^2 n)$ time, where $\kappa=|\calA_{b^*}|\leq n$. Equivalently, $\kappa$ is the minimum number of points of $P$ covered by any half-plane of $H$. We thus obtain the following result. 

\begin{theorem}
Given a set $P$ of $n$ weighted points and a set $H$ of $n$ half-planes in the plane, one can compute in $O(n^{5/2}\log^2 n)$ time a minimum-weight subset of $P$ as a hitting set for $H$. More specifically, the runtime is $O(\kappa\cdot n^{{3}/{2}} \log^2 n)$, where $\kappa$ is the minimum number of points of $P$ covered by any half-plane of $H$. 
\end{theorem}

\section{Implementation of Algorithm~\ref{algo:10}}
\label{sec:implesketch}

In this section, we discuss the implementation of Algorithm~\ref{algo:10}. Since it is quite similar to the algorithm for unit disks in \cite{ref:LiuOn24}, we only sketch the main algorithm framework here. We refer the reader to the detailed discussion and analysis in \cite{ref:LiuOn24}.

To simplify the notation, we temporarily let $P=P_{\alpha}$ and $H=H_{\alpha}$ in the rest of this section. 
Let $L$ denote the set of the bounding lines of the half-planes of $H$. Again to simplify the notation, we let $n=|P|=|H|=|L|$. 

\paragraph{\bf Cuttings.}
Our algorithm will construct a cutting on the lines of $L$~\cite{ref:ChazelleCu93}. 
For a parameter $r$ with $1 \leq r \leq n$, a {\em $(1/r)$-cutting} $\Xi$ for $L$ is a collection of cells (each of which is a triangle) whose union covers the entire plane and whose interiors are pairwise disjoint such that
$|L_{\sigma}| \leq n/r$, where $L_{\sigma}$ is the subset of lines of $L$ that cross the interior of $\sigma$ ($L_{\sigma}$ is called the {\em conflict list} of $\sigma$).
Let $H_{\sigma}$ be the set of the half-planes of $H$ whose bounding lines are in $L_{\sigma}$. 
The {\em size} of $\Xi$ is the number of its cells. 

We say that a cutting $\Xi'$ \emph{$c$-refines} another cutting $\Xi$ if each cell of $\Xi'$ is completely inside a cell of $\Xi$ and each cell of $\Xi$ contains at most $c$ cells of $\Xi'$. Let $\Xi_0$ denote the cutting with a single cell that is the entire plane. We define cuttings $\{\Xi_0, \Xi_1, ..., \Xi_k\}$, in which each $\Xi_i$, $1 \leq i \leq k$, is a $(1/\rho^i)$-cutting of size $O(\rho^{2i})$ that $c$-refines $\Xi_{i - 1}$, for two constants $\rho$ and $c$. By setting $k = \lceil \log_\rho r \rceil$, the last cutting $\Xi_k$ is a $(1/r)$-cutting. The sequence $\{\Xi_0, \Xi_1, ..., \Xi_k\}$ is called a {\em hierarchical $(1/r)$-cutting} for $L$. If a cell $\sigma'$ of $\Xi_{i - 1}$, $1 \leq i \leq k$, contains a cell $\sigma$ of $\Xi_i$, we say that $\sigma'$ is the \emph{parent} of $\sigma$ and $\sigma$ is a \emph{child} of $\sigma'$.  We can also define {\em ancestors} and {\em descendants} correspondingly.
As such, the hierarchical $(1/r)$-cutting forms a tree structure with the single cell of $\Xi_0$ as the root. We often use $\Xi$ to denote the set of all cells in all cuttings $\Xi_i$, $0\leq i\leq k$.
The total number of cells of $\Xi$ is $O(r^2)$~\cite{ref:ChazelleCu93}. 
A hierarchical $(1/r)$-cutting for $L$ can be computed in $O(nr)$ time~\cite{ref:ChazelleCu93}, along with the conflict lists $L_{\sigma}$ (and thus $H_{\sigma}$) for all cells $\sigma\in \Xi$. 


In what follows, we first discuss a preprocessing step in Section~\ref{sec:preprocess}. The algorithms for handling the two key operations are described in the subsequent two subsections, respectively. Section~\ref{sec:summary} finally summarizes everything.

\subsection{Preprocessing}
\label{sec:preprocess}
We compute a hierarchical $(1/r)$-cutting $\{\Xi_0, \Xi_1, ..., \Xi_k\}$ for $L$ in $O(nr)$ time~\cite{ref:ChazelleCu93,ref:WangUn23}, for a parameter $1\leq r\leq n$ to be determined later. Let $\Xi$ denote the set of cells in all these cuttings. 
Using the conflict lists $L_{\sigma}$, for each half-plane $h_i\in H$, we compute a list $\phi(h_i)$  of all cells $\sigma\in \Xi$ such that $h_i\in L_{\sigma}$. This can be done in $O(nr)$ time. 

For any region $R$ in the plane, let $P(R)$ denote the subset of points of $P$ that are inside $R$. 

We compute the subset $P(\sigma)$ of all cells $\sigma$ in the last cutting $\Xi_k$. This can be done by a point location procedure in $O(n\log r)$ time. 


Note that a cell $\sigma\in \Xi$ is the ancestor of another cell $\sigma'\in \Xi$ (alternatively, $\sigma'$ is a descendant of $\sigma$) if and only if $\sigma$ fully contains $\sigma'$. 
For convenience, we consider $\sigma$ an ancestor of itself but not a descendant of itself.  
Let $X(\sigma)$ denote the set of all ancestors of $\sigma$ and $Y(\sigma)$ the set of all descendants of $\sigma$.
Hence, $\sigma$ is in $X(\sigma)$ but not in $Y(\sigma)$. Let $Z(\sigma)$ denote the set of all children of $\sigma$

For each point $p\in P$, we associate with it a variable $\lambda(p)$. 
For each cell $\sigma\in \Xi$, we associate with it two variables: $minCost(\sigma)$ and $\lambda(\sigma)$.
If $|P(\sigma)|=\emptyset$, then we let $minCost(\sigma)=\infty$ and $\lambda(\sigma)=0$. 
Our algorithm for handling the two key operations will maintain the following two invariants. 
\begin{invariant}
For any point $p\in P$, $cost(p)=w(p)+\lambda(p)+\sum_{\sigma'\in X(\sigma)}\lambda(\sigma')$, where $\sigma$ is the cell of $\Xi_k$ that contains $q$. 
\end{invariant}
\begin{invariant}
For each cell $\sigma\in \Xi$ with $P(\sigma)\neq \emptyset$, if $\sigma$ is a cell of $\Xi_k$, then $minCost(\sigma)=\min_{q\in P(\sigma)}(w(p)+\lambda(p))$; otherwise, $minCost(\sigma)=\min_{\sigma'\in Z(\sigma)}(minCost(\sigma')+\lambda(\sigma'))$. 
\end{invariant}

The above algorithm invariants further imply the following observation. 
\begin{observation}\label{obser:invariant}
For each cell $\sigma\in \Xi$ with $P(\sigma)\neq \emptyset$, $\min_{p\in P(\sigma)}cost(p)=minCost(\sigma)+\sum_{\sigma'\in X(\sigma)}\lambda(\sigma')$.    
\end{observation}

For each cell $\sigma\in \Xi$, we also maintain $\calL(\sigma)$, a list comprising all descendant cells $\sigma'$ of $\sigma$ with $\lambda(\sigma')\neq 0$ and all points $p\in P(\sigma)$ with $\lambda(p)\neq 0$. As $\calL(\sigma)$ has both cells of $Y(\sigma)$ and points of $P(\sigma)$, for convenience, we use an ``element'' to refer to either a cell or a point of $\calL(\sigma)$.
As will be seen later, whenever the algorithm sets $\lambda(\sigma)$ to a nonzero value for a cell $\sigma\in \Xi$,  $\sigma$ will be added to $\calL(\sigma')$ for every ancestor $\sigma'$ of $\sigma$ with $\sigma'\neq \sigma$. Similarly, whenever the algorithm sets $\lambda(p)$ to a nonzero value for a point $p$, then $p$ will be added to $\calL(\sigma')$ for every cell $\sigma'\in X(\sigma)$, where $\sigma$ is the cell of $\Xi_k$ containing $p$. 


In the beginning, we initialize the variables $\calL(\cdot)$, $\lambda(\cdot)$, $minCost(\cdot)$ so that the algorithm invariants hold. 
For each cell $\sigma\in \Xi$, we set $\calL(\sigma)=\emptyset$ and $\lambda(\sigma)=0$. For each point $p\in P$, we set $\lambda(p)=0$. Since $cost(p)=w(p)$ initially, Algorithm Invariant 1 holds. 

We next set $minCost(\sigma)$ for all cells of $\sigma\in \Xi$ in a bottom-up manner following the tree structure of $\Xi$. Specifically, for each cell $\sigma\in \Xi_k$, set $minCost(\sigma)=\min_{p\in P(\sigma)}w(p)$. If $P(\sigma)=\emptyset$, we set $minCost(\sigma)=\infty$. 
Then, we set $minCost(\sigma)$ for all cells of $\sigma\in \Xi_{k-1}$ with $minCost(\sigma)=\min_{\sigma'\in Z(\sigma)}(minCost(\sigma')+\lambda(\sigma'))$. We continue this process to set $minCost(\sigma)$ for cells in $\Xi_{k-2},\Xi_{k-3},\ldots,\Xi_0$. 
This establishes Algorithm Invariant 2. 

In addition, for each cell $\sigma$ in the last cutting $\Xi_k$, we construct a min-heap $\calH(\sigma)$ on all points $p$ of $P(\sigma)$ with the values $w(p)+\lambda(p)$ as ``keys''. Using the heap, if $\lambda(p)$ changes for a point $p\in P(\sigma)$, $minCost(\sigma)$ can be updated in $O(\log n)$ time. 

This finishes our preprocessing step, which takes $O(n\log n+nr)$ time in total.

\subsection{The \findmin\ operation}

Consider a half-plane $h_i$ in \findmin\ operation of the $i$-th iteration of Algorithm~\ref{algo:10}. The goal is to compute $\min_{p\in P(h_i)}cost(p)$, i.e., the minimum cost of all points of $P$ inside the half-plane $h_i$ (note that since we use $P$ to represent $P_{\alpha}$, the notation $P_{\alpha}(h_i)$ in Algorithm~\ref{algo:10} becomes $P(h_i)$). 

Recall that $\phi(h_i)$ is the list of all cells $\sigma\in \Xi$ such that $h_i\in H_{\sigma}$.
Define $\phi_1(h_i)$ to be the set of all cells of $\phi(h_i)$ that are from $\Xi_k$ and let $\phi_2(h_i)=\phi(h_i)\setminus \phi_1(h_i)$. Define $\phi_3(h_i)$ to be the set of cells $\sigma\in \Xi$ such that $\sigma$'s parent is in $\phi_2(h_i)$ and $\sigma$ is completely contained in $h_i$. 
The following observation is due to the definition of the hierarchical cutting. 

\begin{observation}\label{obser:Qunion}
$P(h_i)$ is the union of $\bigcup_{\sigma\in \phi_1(h_i)}(P(\sigma)\cap h_i)$ and $\bigcup_{\sigma\in \phi_3(h_i)}P(\sigma)$.
\end{observation}

With Observation~\ref{obser:Qunion}, our algorithm for \findmin\ works as follows. 
Let $\xi$ be a variable, which is $\infty$ initially. At the end of the algorithm, we will have $\xi=\min_{p\in P(h_i)}cost(p)$. For each cell $\sigma\in \phi(h_i)$, we process it as follows. 

\begin{itemize}
\item 
If $\sigma\in \phi_1(h_i)$, for each point $p\in P(\sigma)$, by Algorithm Invariant 1, we have $cost(p)=w(p)+\lambda(p)+\sum_{\sigma'\in X(\sigma)}\lambda(\sigma')$. If $p\in h_i$, we first compute $cost(p)$ by visiting all cells of $X(\sigma)$, which takes $O(\log r)$ time, and then update $\xi=\min\{\xi,cost(p)\}$. 

\item 
If $\sigma\in \phi_2(h_i)$, for each child $\sigma'$ of $\sigma$ that is fully contained in $h_i$ (i.e., $\sigma\in \phi_3(h_i)$), we compute $\xi_{\sigma'}=minCost(\sigma')+\sum_{\sigma''\in X(\sigma')}\lambda(\sigma'')$ by visiting all cells of $X(\sigma')$, which takes $O(\log r)$ time. By Observation~\ref{obser:invariant}, we have $\xi_{\sigma'}=\min_{p\in P(\sigma)}cost(p)$. Then we update $\xi=\min\{\xi,\xi_{\sigma'}\}$. 
\end{itemize}

After processing all cells $\sigma\in \phi(h_i)$ as above, we return $\xi$, which is equal to $\min_{p\in P(h_i)}cost(p)$. This finishes the \findmin\ operation. The total time of the 
\findmin\ operations in the entire Algorithm~\ref{algo:10} is bounded by $O((nr+n^2/r)\log r)$.

\subsection{The \reset\ operation}
Consider the \reset\ operation in the $i$-th iteration of Algorithm~\ref{algo:10}. The goal is to reset $cost(p)=w(p)+\delta_i$ for all points $p\in P(\overline{h_i})$. To this end, we need to update the $\lambda(\cdot)$ and $minCost(\cdot)$ values for certain cells of $\Xi$ and points of $P$ so that the algorithm invariants still hold. 

Define $\phi_4(h_i)$ as the set of cells $\sigma\in \Xi$ such that $\sigma$'s parent is in $\phi_2(h_i)$ and $\sigma$ is completely outside $h_i$. Let $\overline{h_i}$ denote the complementary half-plane of $h_i$. We have the following observation, which is analogous to Observation~\ref{obser:Qunion}.

\begin{observation}\label{obser:resetunion}
$P(\overline{h_i})$ is the union of $\bigcup_{\sigma\in \phi_1(h_i)}(P(\sigma)\cap \overline{h_i})$ and $\bigcup_{\sigma\in \phi_4(h_i)}P(\sigma)$.
\end{observation}

Our algorithm for \reset\ works as follows. Consider a cell $\sigma\in \phi(h_i)$. Depending on whether $\sigma$ is from $\phi_1(h_i)$ or $\phi_2(h_i)$, we process it in different ways. 

If $\sigma$ is from $\phi_1(h_i)$, we process $\sigma$ as follows. For each point $p\in P(\sigma)$, if $p\in \overline{h_j}$, then we are supposed to reset $cost(p)$ to $w(p)+\delta_i$. To achieve the effect, we do the following. First, we set $\lambda(p)=\delta_i-\sum_{\sigma'\in X(\sigma)}\lambda(\sigma')$. 
After that, we have $w(p)+\lambda(p)+\sum_{\sigma'\in X(\sigma)}\lambda(\sigma')=w(p)+\delta_i$, which establishes the first algorithm invariant for $p$. For the second algorithm invariant, we first update $minCost(\sigma)$ using the heap $\calH(\sigma)$, i.e., by updating the key of $p$ to the new value $w(p)+\lambda(p)$. Next, we update $minCost(\sigma')$ for all ancestors $\sigma'$ of $\sigma$ in a bottom-up manner using the formula $minCost(\sigma')=\min_{\sigma''\in Z(\sigma')}(minCost(\sigma'')+\lambda(\sigma''))$. Since each cell has $O(1)$ children, updating all ancestors of $\sigma$ takes $O(\log r)$ time. This establishes the second algorithm invariant. Finally, since $\lambda(p)$ has just been changed, if $\lambda(p)\neq 0$, then we add $p$ to the list $\calL(\sigma')$ for all cells $\sigma'\in X(\sigma)$. 
This finishes the processing of $p$, which takes $O(\log r+\log n)$ time. 
Processing all points of $p\in P(\sigma)$ as above takes $O(|P(\sigma)|\cdot (\log r+\log n))$ time. 

If $\sigma$ is from $\phi_2(h_i)$, then we process $\sigma$ as follows. For each child $\sigma'$ of $\sigma$, if $\sigma'$ is completely outside $h_i$, then we process $\sigma'$ as follows. We are supposed to reset $cost(p)$ to $w(p)+\delta_i$ for all points $p\in P(\sigma')$. In other words, the first algorithm invariant does not hold any more and we need to update our data structure to restore it. Note that the second algorithm invariant still holds. 
To achieve the effect, 
for each element $e$ in $\calL(\sigma')$ (recall that $e$ is either a cell of $Y(\sigma')$ or a point of $P(\sigma')$), we process $e$ as follows. First, we remove $e$ from $\calL(\sigma')$. Then we reset $\lambda(e)=0$. If $e$ is a point of $P(\sigma')$, then let $\sigma_e$ be the cell of $\Xi_k$ that contains $e$; otherwise, $e$ is a cell of $Y(\sigma')$ and let $\sigma_e$ be the parent of $e$. Since $\lambda(e)$ is changed, we update $minCost(\sigma'')$ for all cells $\sigma''\in X(\sigma_e)$ in the same way as above in the first case for processing $\phi_1(h_i)$. 
This finishes processing $e$, after which the second algorithm invariant still holds. After all elements of $\calL(\sigma')$ are processed as above, $\calL(\sigma')$ becomes $\emptyset$ and we reset $\lambda(\sigma')=\delta_i-\sum_{\sigma''\in X(\sigma')\setminus\{\sigma'\}}\lambda(\sigma'')$. Since $\lambda(\sigma')$ has been changed, we update $minCost(\sigma'')$ for all cells $\sigma''\in X(\sigma)$ in the same way as before, after which the second algorithm invariant still holds. In addition, if $\lambda(\sigma')\neq 0$, then we add $\sigma'$ to the list $\calL(\sigma'')$ for all cells $\sigma''\in X(\sigma)$. 
This finishes processing $\sigma'$, which takes $O(|\calL(\sigma')|\cdot (\log r+\log n))$ time. The first algorithm invariant is established for all points $p\in P(\sigma')$.     

This finishes the \reset\ operation. The total time of the \reset\ operations in the entire Algorithm~\ref{algo:10} is bounded by $O((nr+n^2/r)\cdot \log r\cdot (\log r+\log n))$. 

\subsection{Summary}
\label{sec:summary}

According to the above discussion, the total time of the overall algorithm is $O(n\log n+n\log r+(nr+n^2/r)\cdot \log r\cdot (\log n+\log r))$. Recall that $1\leq r\leq n$. Setting $r=\sqrt{n}$ gives the upper bound $O(n^{3/2}\log^2 n)$ for the overall time complexity of Algorithm~\ref{algo:10}.

\bibliographystyle{plainurl}

\end{document}